\documentclass[11pt]{article}
\usepackage[utf8]{inputenc}

\usepackage{tikz}
\usetikzlibrary{positioning}
\usetikzlibrary{fit}

\usepackage{adjustbox}

\usepackage{amsmath}
\usepackage{amssymb}
\usepackage{amsthm}
\usepackage{authblk}
\usepackage{natbib}

\usepackage{algpseudocode}
\usepackage{algorithm}
\usepackage{tikz}
\usetikzlibrary{positioning}

\usepackage{enumitem}

\usepackage{fullpage}

\usepackage{xcolor}

\usepackage{hyperref}
\usepackage[all]{hypcap}

\newcommand{\naturalnumberpositive}{\ensuremath{{\mathbb{N}^+}}}
\newcommand{\naturalnumberzero}{\ensuremath{{\mathbb{N}_0}}}

\newcommand{\p}{\ensuremath{{\rm P}}}
\newcommand{\np}{\ensuremath{{\rm NP}}}

\title{A Critique of Quigley's ``A Polynomial Time Algorithm for 3SAT''\thanks{Supported in part by NSF grant CCF-2006496}}

\author{Nicholas~DeJesse}
\author{Spencer~Lyudovyk}
\author{Dhruv~Pai}

\affil{Department of Computer Science\\University of Rochester\\Rochester, NY 14627, USA}

\newtheorem{example}{Example}

\newtheorem{theorem}{Theorem}
\newtheorem{definition}{Definition}

\newtheorem{lemma}[theorem]{Lemma}

\date{October 27, 2025}

\begin{document}\sloppy

\maketitle

\begin{abstract}
In this paper, we examine Quigley's ``A Polynomial Time Algorithm for 3SAT" \cite{qui:poly-sat}. Quigley claims to construct an algorithm that runs in polynomial time and determines whether a boolean formula in 3CNF form is satisfiable. Such a result would prove that 3SAT~$\in$~$\p$ and thus $\p=\np$. We show Quigley's argument is flawed by providing counterexamples to several lemmas he attempts to use to justify the correctness of his algorithm. We also provide an infinite class of 3CNF formulas that are unsatisfiable but are classified as satisfiable by Quigley's algorithm. In doing so, we prove that Quigley's algorithm fails on certain inputs, and thus his claim that $\p=\np$ is not established by his paper.
\end{abstract}

\section{Introduction}
This critique analyzes Quigley's ``A Polynomial Time Algorithm for 3SAT" \cite{qui:poly-sat}, which claims to provide a polynomial time algorithm for determining whether a boolean formula in 3CNF form is satisfiable \footnote{We critique Version 1 of Quigley's paper, which at the time of writing is the most recent and only version.}. Since 3SAT is NP-complete (see \cite{sip:b:introduction-third-edition}), this result implies that $\p=\np$.

The relationship between the complexity classes $\p$ and $\np$ is one of the most important unsolved problems in complexity theory, and proving equality (or lack thereof) would have drastic effects on computer science and many other fields. For example, if $\p=\np$, many cryptographic systems, such as RSA encryption, would be insecure, compromising the security of communication on the Internet \cite{sip:b:introduction-third-edition}. Additionally, proving $\p=\np$ would have major benefits in operations research and logistics, since many problems previously thought to be intractable, such as the Traveling Salesman Problem, which is $\np$-complete \cite{cor-lei-riv-ste:b:algorithms-third-edition}, would become solvable in polynomial time. Conversely, a proof of $\p\neq\np$ would have major consequences in the field of computational complexity, such as implying the existence of $\np$-intermediate languages (i.e. $\np$ languages that are neither in $\p$ nor $\np$-complete) \cite{lad:j:np-incomplete}.

In this paper, we argue that Quigley's algorithm, which claims to decide instances of 3SAT deterministically in polynomial time \cite{qui:poly-sat}, is flawed by providing counterexamples to the lemmas with which he attempts to argue his algorithm's correctness, and we provide an infinite set of unsatisfiable boolean formulas that Quigley's algorithm classifies as satisfiable.

\section{Preliminaries}
Let $\naturalnumberzero=\{0, 1, 2, 3,\ldots\}$ and $\naturalnumberpositive = \{1, 2, 3, \ldots\}$. A terminal (i.e., variable) is a symbol appearing in a boolean formula that can be assigned a value of either true or false. A term (i.e., literal) is a terminal in either positive or negated form that appears in a clause. For example, in the formula $x_1 \lor \overline{x_1} \lor x_2$, the two terminals are $x_1$ and $x_2$, and the three terms are $x_1$, $\overline{x_1}$, and $x_2$. A boolean formula is in $k$CNF form for some $k \in \naturalnumberpositive$ if it is a conjunction of any number of clauses, where each clause is a disjunction of at most $k$ terms. A partial assignment to a boolean formula maps each terminal of a proper subset of the terminals in the formula to either true or false, whereas a complete assignment maps every terminal in the formula to either true or false. A satisfying assignment to a boolean formula is an assignment that, when applied to the terminals in a boolean formula, causes the boolean formula to evaluate to true.
A boolean formula is satisfiable if and only if there exists a satisfying assignment for it, and it is unsatisfiable if such a satisfying assignment does not exist.

Additionally, we assume the reader is familiar with nondeterminism, the complexity classes $\p$ and $\np$, and big O notation. For more information on any of these topics, readers can consult any standard textbook \citep[e.g.,][]{aro-bar:computational-complexity, hop-ull:b:automata, sip:b:introduction-third-edition}.

\subsection{Quigley's Definitions}\label{quig-def}
Quigley begins his paper with some basic definitions required to understand both his algorithm and the lemmas with which he attempts to prove its correctness. In this section, we mention the most relevant definitions. Quigley first defines what it means for a clause in a boolean formula in conjunctive normal form, which he refers to as an ``instance," to ``block" an assignment to the terminals in the formula. We provide Quigley's definition as stated in his paper and then clarify some of the terminology.
\begin{definition}[{\cite[Definition 3.1, p.~3]{qui:poly-sat}}]
    An assignment, $A$, is said to be blocked by a clause, $C$ if, given an instance containing $C$, there is no way that $A$ allows $C$ to evaluate to True, and thus there is no way $A$ allows the instance to evaluate to True.
\end{definition}

In other words, given a boolean formula in conjunctive normal form (i.e., an instance) containing some clause $C$, an assignment $A$ to the variables in the formula is blocked by the clause $C$ if $A$ makes the instance evaluate to false by making $C$ evaluate to false. This means a clause $C$ blocks an assignment $A$ if and only if $C$ evaluates to false under the assignment $A$. Quigley then uses this definition to define clause implication.

\begin{definition}[{\cite[Definition 3.2, p.~3]{qui:poly-sat}}] 
    A clause, $C$, is said to imply another clause, $D$, if all assignments blocked by $D$ are also blocked by $C$.
\end{definition}

This means that a clause $C$ in a boolean formula implies a clause $D$ if and only if for any complete assignment $A$ to the terminals in the formula, if $C$ evaluates to true under $A$, then $D$ evaluates to true under $A$.

\section{Analysis of Quigley's Arguments}

\subsection{Quigley's Rules for Implication}

Unless otherwise specified, when we refer to clauses as being ``implied" by other clauses, we will use Quigley's definition of implication.

Quigley provides several rules that his algorithm uses to find new clauses that are implied by existing clauses in a given boolean formula in 3CNF form. First, he defines ``expansion."

\begin{lemma}[{\cite[Lemma 5.8, p.~6]{qui:poly-sat}}]
    Given a clause, $C$, and a terminal, $t$, that's not in $C$, then two new clauses can be implied consisting of all the terms of $C$ appended to either the positive form of $t$ or the negated form of $t$.
\end{lemma}

With this lemma, Quigley defines the rule of ``expansion." For instance, this rule can be used to ``expand" an existing clause
\[ x_1 \lor \cdots \lor x_n \]
with some distinct terms $x_1, ..., x_n$
into two new clauses
\[ x_1 \lor \cdots \lor x_n \lor t \]
and
\[ x_1 \lor \cdots \lor x_n \lor \overline{t} \]
for some terminal $t$ not found in the original clause.

To prove that the new clauses are implied, let $C$ and $D$ be clauses such that the set of terms in $C$ is a subset of the set of terms in $D$. Now, let $A$ be an assignment to the terminals in $D$. If $D$ evaluates to false under this assignment, then every term in $D$ evaluates to false under $A$ since $D$ is a disjunction of terms. Therefore, since the set of terms in $C$ is a subset of the set of terms in $D$, $C$ must also evaluate to false because every term in $C$ evaluates to false. This means that every assignment blocked by $D$ is also blocked by $C$, so $C$ implies $D$.

Furthermore, Quigley defines a second rule similar to resolution in propositional logic.

\begin{lemma}[{\cite[Lemma 5.9, p.~6]{qui:poly-sat}}]\label{lemma-5.9}
    If two clauses share the same terminal, $t$, such that $t$ is positive in one clause and negated in the other, then these clauses imply a new clause which is composed of all the terms in both clauses except terms containing $t$.
\end{lemma}

With this lemma, Quigley defines a rule that is nearly identical to the definition of resolution in propositional logic, which is known to be sound (see \cite{nor-rus:aima}). In Quigley's proof of this lemma \cite{qui:poly-sat}, he states that if a term appears in both of the original clauses, the new clause contains that term exactly once. In other words, the new clause does not contain any duplicate terms. Note that if an implied clause contains two contradicting terms $x$ and $\overline{x}$ for some terminal $x$, then the clause evaluates to true under any assignment.

For instance, consider two clauses
\[ x_1 \lor \cdots \lor x_n \lor y_1 \lor \cdots \lor y_m \lor t \]
and
\[ x_1 \lor \cdots \lor x_n \lor z_1 \lor \cdots \lor z_k \lor \overline{t} \]
for some terms $x_1, ..., x_n, y_1, ..., y_m, z_1, ..., z_k$ and some terminal $t$ such that $x_1, ..., x_n, y_1, ..., y_m, t$ are all distinct, $x_1, ..., x_n, z_1, ..., z_k, \overline{t}$ are all distinct, and $x_1, ..., x_n$ are the only common terms between the two clauses. Then, using the rule of resolution, these two clauses imply a new clause
\[ x_1 \lor \cdots \lor x_n \lor y_1 \lor \cdots \lor y_m \lor z_1 \lor \cdots \lor z_k. \]

Quigley uses these rules of expansion and resolution to find new clauses that are implied by existing clauses in a 3CNF formula, and he bases his algorithm on these rules. We provide this algorithm in the following section.

Quigley also provides bounds on the minimum and maximum lengths of implied clauses, which are given below.

\begin{lemma}[{\cite[Lemma 5.10, p.~7]{qui:poly-sat}}]
    Given two clauses of lengths $k$ and $m$ that imply another clause by (Quigley's) Lemma 5.9, the length of the implied clause will fall in the range $\max(k, m) - 1$ to $k + m - 2$ where $\max(k, m)$ represents the parameter with the greatest value.
\end{lemma}

Quigley argues that the smallest clause that can be implied by a pair of clauses with lengths $k$ and $m$ is of length $\max(k,m)-1$. This occurs when all but one of the terms in one clause appears in the other. In this case, the implied clause contains all the terms from the larger clause except the one term that is removed during the resolution process. Since duplicate terms in implied clauses are removed, the terms that are shared between the two initial clauses appear only once in the implied clause. Similarly, Quigley states that the largest clause that can be implied by two clauses of lengths $k$ and $m$ has length $k + m - 2$. This occurs when no terms are shared between the two clauses, so the implied clause contains all but one term from each clause.

We found no errors in Quigley's proofs of his Lemmas 5.8, 5.9, and 5.10.

\subsection{Analysis of Quigley's Algorithm}\label{sec-analysis-of-alg}

In this section, we provide Quigley's algorithm, which he claims correctly classifies any instance of a 3CNF formula as satisfiable or unsatisfiable. Quigley's algorithm repeatedly iterates through the clauses of a 3CNF boolean formula to determine whether or not the formula is satisfiable. The algorithm constructs new clauses via his rules of expansion and resolution that are ``added" to the instance by conjoining the instance and the new clause with the logical AND operator. Note that the old clauses that are used to imply new clauses are still kept in the formula. We first provide Quigley's algorithm as stated in his paper, with minor adjustments in punctuation. We will later clarify more precisely how we believe his algorithm works.

Quigley's algorithm, as stated in his paper, is as follows \cite[p.~20]{qui:poly-sat}.

\begin{enumerate}
    \item For each clause in the instance, $C$, of length $3$ or less:
    \begin{enumerate}
        \item For each clause in the instance, $D$, of length $3$ or less:
        \begin{enumerate}
            \item Get all clauses implied by $C$ and $D$ according to (Quigley's) Lemma 5.9 and add them to the instance.
            \item Check if this new clause is in the instance and update a flag accordingly.
        \end{enumerate}
        \item Expand $C$ to get all possible clauses with a maximum length of $3$ and add them to the instance.
        \item For each new clause from the previous step:
        \begin{enumerate}
            \item Check if the new clause is in the instance.
        \end{enumerate}
    \end{enumerate}
    \item For each clause in the instance, $E$, of length $1$:
    \begin{enumerate}
        \item For each clause in the instance, $F$, of length $1$:
        \begin{enumerate}
            \item If $E$ and $F$ contain the same terminal in which it is positive in one clause and negated in the other, the clauses are contradicting and the instance is unsatisfiable, end.
        \end{enumerate}
    \end{enumerate}
    \item Repeat ($1$)-($2$) until no new clauses are added.
    \item If it reaches here, the instance is satisfiable, end.
\end{enumerate}

It is unclear whether clauses of length 4 or greater that are implied in step 1.a.i of Quigley's algorithm are added to the instance. Step 1.a.i states that all implied clauses are added to the instance, but Quigley's analysis of the runtime of step $3$ implies that these clauses are discarded since $O(n^3)$ clauses are added to the instance \cite[p.~21]{qui:poly-sat}. We assume implied clauses of length 4 or greater are simply not added to the instance because they will not be iterated over again, per the conditions in steps $1$ and $2$. This assumption only affects the runtime of the algorithm, not its correctness, and therefore does not affect the counterexample we provide in the next section.

It is also unclear whether or not clauses implied during step 1.a.i in each iteration of step 1 are used to imply more clauses in that same iteration of step 1. We assume that any new implied clauses are not iterated over during the same iteration of step $1$ in which they were previously implied because if they were, then all clauses that would be implied in future iterations of step $1$ would be implied during the first iteration. This would mean that step $3$ of Quigley's algorithm is unnecessary. Therefore, we also assume that any clauses implied during step 1.a.i are not iterated over until the next iteration of step 1.a.i. This assumption also does not affect the correctness of the algorithm or our counterexample in the next section.

With these assumptions, we believe Quigley's algorithm works as follows:

\begin{enumerate}
    \item Create a temporary list of clauses $L$, initialized to be empty.
    \item For each clause $C$ of length 3 or less in the instance:
    \begin{enumerate}
        \item For each clause $D$ of length 3 or less in the instance:
        \begin{enumerate}
            \item Let $V$ be the set of all clauses of length $3$ or less implied by $C$ and $D$ according to Quigley's Lemma 5.9. For each clause $v\in V$, if $v$ is not already in the instance and $v$ does not contain any contradicting terms, append $v$ to $L$.
        \end{enumerate}
        \item Expand $C$ to get all clauses implied by $C$ with a maximum length of 3 according to Quigley's Lemma 5.8, and append them to $L$.
    \end{enumerate}
    \item For each clause $l\in L$:
    \begin{enumerate}
        \item Check if $l$ is in the instance. If it is not, add it to the instance.
    \end{enumerate}
    \item For each clause in the instance, $E$, of length 1:
    \begin{enumerate}
        \item For each clause in the instance, $F$, of length 1:
        \begin{enumerate}
            \item If $E$ and $F$ contain the same terminal, which is positive in one clause and negated in the other, then the clauses are contradicting. Return that the instance is unsatisfiable.
        \end{enumerate}
    \end{enumerate}
    \item Repeat steps 1, 2, and 3 until no new clauses are added.
    \item If the algorithm reaches this step, return that the instance is satisfiable.
\end{enumerate}

We include the above description of the algorithm for clarity. However, in the following sections, when we refer to numbered steps of the algorithm, we will be referring to steps of Quigley's algorithm as originally stated in his paper and earlier in this section.

\subsubsection{Runtime of Quigley's Algorithm}\label{section:runtime}

Now, we will examine the runtime of Quigley's algorithm, as presented in his paper \cite[p.~20]{qui:poly-sat}. Quigley claims his algorithm runs in time $O(n^{12})$, which is polynomially bounded. He shows this by providing an upper bound for each step in his algorithm.

To provide these upper bounds, Quigley assumes that any particular clause examined or iterated over by his algorithm contains no repeated terminals. In particular, for any clause, any terminal $x$ found in the clause in the form of the term $x$ or the term $\overline{x}$ cannot be found elsewhere in the clause in either the positive or negated form; for instance, if the term $\overline{x}$ is in the clause, a duplicate of this term (i.e., $\overline{x}$) cannot be found elsewhere in the clause, and the term $x$ cannot be found in the clause. Note that this assumption applies only within clauses, so terminals may repeat between different clauses. Additionally, since the order of the terms in a clause does not affect the clause's truth value or its satisfiability, Quigley considers clauses containing the same terms in different orders to be the same clause.

In general, these assumptions may not be satisfied by some boolean formulas. In these cases, the polynomial runtime bound is not necessarily guaranteed simply because the input itself may have too many clauses and thus be too large to be polynomially bounded in the number of variables. However, any formula that does not satisfy these assumptions can be converted to an equivalent formula that does satisfy them. To do so, one can remove from the formula any clauses that contain both the term $x$ and the term $\overline{x}$ for some terminal $x$ since such clauses always evaluate to true and thus do not affect the satisfiability of the overall formula. Also, one can remove duplicates of the same term in each clause (e.g., if a term $x$ is found more than once in some clause, only one occurrence of $x$ needs to be included in that clause, so all other occurrences can be removed). Finally, one can remove repeated clauses that contain the same terms in different orders (e.g., if one clauses is $x_1 \lor x_2$ and another clause is $x_2 \lor x_1$, only one of these clauses needs to be included in the overall formula). Further, if the initial input satisfies the assumptions, then by following these same guidelines (i.e., not storing clauses that trivially evaluate to true due to containing a terminal in both positive and negated form, not storing duplicate terms in each clause, and not storing clauses that contain the same terms as some other clause already present in the formula) during the execution of Quigley's algorithm, one can guarantee that all future clauses implied by Quigley's algorithm, and thus all clauses iterated over by his algorithm, also satisfy the assumptions. Thus, to allow for a more meaningful and thorough analysis of Quigley's algorithm, we will assume inputs to his algorithm satisfy the assumptions.

Now, consider possible formulas satisfying the assumptions above that can be constructed using the terminals $x_1, ..., x_n$ for some $n \in \naturalnumberpositive$. Note that for any $n \in \naturalnumberpositive$ with $k \leq n$, there are $\binom{n}{k} \cdot 2^k$ distinct boolean clauses with exactly $k$ terms, in which the order of the terms does not matter and no terminal can be repeated more than once (as described previously), that can be constructed from $n$ distinct variables. This is because there are $\binom{n}{k}$ ways to choose $k$ of the $n$ variables and there are 2 ways for each of these $k$ variables to appear (i.e., either positive or negated).
Thus, a boolean formula with $n$ variables and clauses of length at most 3 has at most
\[ \binom{n}{3} \cdot 2^3 + \binom{n}{2} \cdot 2^2 + \binom{n}{1} \cdot 2^1 = O(n^3) \]
clauses. As such, steps $1$ and 1.a of Quigley's algorithm each take time $O(n^3)$ in the worst case as they must each iterate over all clauses in the boolean formula. Step 1.a.i takes $O(1)$ time since clauses cannot exceed length 3, so when resolving two such clauses, there are a constant number of terminals to compare. Step 1.a.ii iterates over all clauses and thus takes $O(n^3)$ time. Next, any clause of length $1$ can be expanded to $O(n^2)$ new clauses of length at most 3, and any clause of length $2$ can be expanded to $O(n)$ new clauses of length at most 3, so step 1.b takes $O(n^2)$ time. Then, step 1.c iterates over the $O(n^2)$ new clauses from step 1.b, and step 1.c.i compares each of them to the $O(n^3)$ other clauses in the formula. As such, the total time complexity of step $1$, including all sub-steps, is on the order of
\[ n^3 \cdot (n^3 \cdot (1 + n^3) + n^2 + n^2 \cdot n^3) = O(n^9). \]
Next, steps 2 and 2.a each iterate over $O(n^3)$ clauses, and step 2.a.i compares two clauses of length $1$ in $O(1)$ time, so the total time complexity of step $2$, including all sub-steps, is on the order of 
\[ n^3 \cdot n^3 \cdot 1 = O(n^6). \]
Step $3$ repeats steps $1$ and $2$ until no new clauses are added. Since there are a maximum of $O(n^3)$ new clauses that can be added in steps $1$ and $2$ (i.e., all possible clauses of length at most 3 formed from $n$ variables), step $3$ causes at most $O(n^3)$ repetitions (which, for instance, can occur when only $1$ new clause is added during each iteration). Finally, step 4 simply ends the algorithm and takes $O(1)$ time.

Based on the time complexity of each individual step, the total time complexity of Quigley's algorithm is the number of repetitions caused by step 3 times the sum of the time complexity of one iteration of each of steps 1 and 2, plus the time complexity of step 4. Thus, the overall time complexity is on the order of
\[ n^3 \cdot (n^9 + n^6) + 1 = O(n^{12}). \]

Under the assumptions stated earlier in this section, we find no issues with Quigley's time complexity analysis. However, even under these assumptions, we find that Quigley's algorithm does not always classify 3CNF formulas correctly. In particular, some unsatisfiable formulas are classified as satisfiable, which we prove in a later section.

\subsection{Analysis of Quigley's Lemmas}

To justify the correctness of his algorithm and provide a bound on its runtime, Quigley introduces several lemmas. In this section, we analyze the lemmas that are most relevant to understanding his algorithm and provide counterexamples. The first of these, Quigley's Lemma 5.11, is given below and has been rephrased and condensed slightly for clarity.

\begin{lemma}[{\cite[Lemma 5.11, p.~7]{qui:poly-sat}}]
    Consider some $k \in \mathbb{N}^+$ with $k \geq 2$. Consider clauses $A$, $B$, and $C$ of length less than $k$, a clause $D$ of length $k$ or $k-1$, and a clause $E$ of length $k$ such  that $A$ and $B$ imply $E$ by (Quigley's) Lemma 5.9 and $C$ and $E$ imply $D$ by (Quigley's) Lemma 5.9. Then, $A$, $B$, and $C$ can imply $D$ by processing only clauses with a maximum length of $k-1$.
\end{lemma}

This lemma states that for any $k\in\naturalnumberpositive$ with $k\geq 2$, the given combination of clauses $A$, $B$, and $C$ meeting certain criteria can imply some other clause $D$ by using only clauses of length less than $k$ as intermediate clauses. However, consider the following counterexample to this lemma. Let $k = 4$, and let $a_1, a_2, a_3, a_4, a_5$ be distinct terminals. Let the clause $A$ be $(a_1 \lor a_2 \lor a_3)$, the clause $B$ be $(\overline{a_1} \lor a_4 \lor a_5)$, the clause $C$ be $(\overline{a_1} \lor \overline{a_2} \lor a_4)$, the clause $D$ be $(\overline{a_1} \lor a_3 \lor a_4 \lor a_5)$, and the clause $E$ be $(a_2 \lor a_3 \lor a_4 \lor a_5)$.
Notice that each of the clauses $A$, $B$, and $C$ has length $k-1$, each of the clauses $D$ and $E$ has length $k$, the clauses $A$ and $B$ imply $E$ by resolving using the terminal $a_1$, and the clauses $C$ and $E$ imply $D$ by resolving using the terminal $a_2$. As such, these five clauses satisfy the hypotheses of Quigley's Lemma 5.11. However, there are no clauses of length less than $k$ implied by just $A$, $B$, and $C$, and no two of these clauses directly imply $D$. Thus, in order for $A$, $B$, and $C$ to imply $D$, an intermediate clause of length at least $k$ must be processed, contradicting Quigley's Lemma 5.11.

In fact, we can extend this counterexample to any arbitrary $k \geq 4$. Let $a_1, ..., a_{k+1}$ be distinct terminals. Let the clause $A$ be $(a_1 \lor a_2 \lor \cdots \lor a_{k-1})$, the clause $B$ be $(\overline{a_1} \lor a_4 \lor a_5 \lor \cdots \lor a_{k+1})$, the clause $C$ be $(\overline{a_1} \lor \overline{a_2} \lor a_4 \lor a_5 \lor \cdots \lor a_k)$, the clause $D$ be $(\overline{a_1} \lor a_3 \lor a_4 \lor \cdots \lor a_{k+1})$, and the clause $E$ be $(a_2 \lor a_3 \lor \cdots \lor a_{k+1})$. As in the example with $k = 4$ given earlier, notice that each of the clauses $A$, $B$, and $C$ has length $k-1$, each of the clauses $D$ and $E$ has length $k$, the clauses $A$ and $B$ imply $E$ by resolving the terminal $a_1$, and the clauses $C$ and $E$ imply $D$ by resolving the terminal $a_2$. As such, these five clauses satisfy the hypotheses of Quigley's Lemma 5.11. However, there are no clauses of length less than $k$ implied by just $A$, $B$, and $C$, and no two of these clauses directly imply $D$. Thus, in order for $A$, $B$, and $C$ to imply $D$, an intermediate clause of length at least $k$ must be processed, contradicting Quigley's Lemma 5.11.

Next, consider Quigley's Lemma 5.17, which is given below and has been rephrased slightly for clarity.

\begin{lemma}[{\cite[Lemma 5.17, p.~12]{qui:poly-sat}}]
    Consider some $k \in \mathbb{N}^+$ with $k \geq 2$. Consider clauses $A$, $B$, $C$, and $D$ of length less than $k$, clauses $E$ and $F$ of length $k$, and a clause $G$ of length $k$ or $k-1$ such that $A$ and $B$ imply $E$ by (Quigley's) Lemma 5.9, $C$ and $D$ imply $F$ by (Quigley's) Lemma 5.9, and $E$ and $F$ imply $G$ by (Quigley's) Lemma 5.9. Then, $A$, $B$, $C$, and $D$ can imply $G$ by processing only clauses with a maximum length of $k-1$.
\end{lemma}

This lemma states that for any $k \in \naturalnumberpositive$ with $k \geq 2$, the given combination of clauses $A$, $B$, $C$, and $D$ meeting certain criteria can imply some other clause $G$ using only clauses of length less than $k$ as intermediate clauses. However, consider the following counterexample. Let $k = 4$, and let $a_1, a_2, a_3, a_4, a_5, a_6$ be distinct terminals. Let the clause $A$ be $(a_1 \lor a_2 \lor a_5)$, the clause $B$ be $(a_3 \lor a_4 \lor \overline{a_5})$, the clause $C$ be $(\overline{a_1} \lor a_2 \lor a_6)$, the clause $D$ be $(a_3 \lor a_4 \lor \overline{a_6})$, the clause $E$ be $(a_1 \lor a_2 \lor a_3 \lor a_4)$, the clause $F$ be $(\overline{a_1} \lor a_2 \lor a_3 \lor a_4)$, and the clause $G$ be $(a_2 \lor a_3 \lor a_4)$. Notice that each of the clauses $A$, $B$, $C$, $D$, and $G$ has length $k-1$, each of the clauses $E$ and $F$ has length $k$, the clauses $A$ and $B$ imply $E$ by resolving the terminal $a_5$, the clauses $C$ and $D$ imply $F$ by resolving the terminal $a_6$, and the clauses $E$ and $F$ imply $G$ by resolving the terminal $a_1$. As such, these clauses satisfy the hypotheses of Quigley's Lemma 5.17. However, the only clause of length less than $k$ implied by some combination of the clauses $A$, $B$, $C$, and $D$ is $(a_2 \lor a_5 \lor a_6)$, which is implied by $A$ and $C$ by resolving the terminal $a_1$. There are no further clauses implied by any of $A$, $B$, $C$, $D$, and the new clause $(a_2 \lor a_5 \lor a_6)$, and no two of these clauses directly imply $G$. Thus, in order for $A$, $B$, $C$, and $D$ to imply $G$, at least one intermediate clause of length at least $k$ must be processed, contradicting Quigley's Lemma 5.17.

Finally, consider Quigley's Lemma 5.18, which is given below and has been rephrased slightly for clarity.

\begin{lemma}[{\cite[Lemma 5.18, p.~15]{qui:poly-sat}}]
    Consider some $k \in \mathbb{N}^+$ with $k \geq 2$. Consider clauses $A$, $B$, and $C$ of length less than $k$, clauses $D$ and $E$ of length $k$, and a clause $F$ of length $k$ or $k-1$ such that $A$ and $B$ imply $D$ by (Quigley's) Lemma 5.9, $C$ expands to $E$ by (Quigley's) Lemma 5.8, and $D$ and $E$ imply $F$ by (Quigley's) Lemma 5.9. Then, $A$, $B$, and $C$ can imply $F$ by processing only clauses with a maximum length of $k-1$.
\end{lemma}

This lemma states that for any $k \in \naturalnumberpositive$ with $k \geq 2$, the given combination of clauses $A$, $B$, and $C$ meeting certain criteria can imply some other clause $F$ using only clauses of length less than $k$ as intermediate clauses. However, consider the following counterexample. Let $k = 4$, and let $a_1, a_2, a_3, a_4, a_5$ be distinct terminals. Let the clause $A$ be $(a_1 \lor a_2 \lor a_5)$, the clause $B$ be $(a_3 \lor a_4 \lor \overline{a_5})$, the clause $C$ be $(\overline{a_1} \lor a_3 \lor a_4)$, the clause $D$ be $(a_1 \lor a_2 \lor a_3 \lor a_4)$, the clause $E$ be $(\overline{a_1} \lor a_2 \lor a_3 \lor a_4)$, and the clause $F$ be $(a_2 \lor a_3 \lor a_4)$. Notice that each of the clauses $A$, $B$, $C$, and $F$ has length $k-1$, each of the clauses $D$ and $E$ has length $k$, the clauses $A$ and $B$ imply $D$ by resolving the terminal $a_5$, the clause $C$ expands to $E$ by adding an $a_2$ term, and the clauses $D$ and $E$ imply $F$ by resolving the terminal $a_1$. As such, these clauses satisfy the hypotheses of Quigley's Lemma 5.18. However, there are no clauses of length less than $k$ implied by just $A$, $B$, and $C$. Thus, in order for $A$, $B$, and $C$ to imply $F$, at least one clause of length $k$ must be processed, therefore contradicting Quigley's Lemma 5.18.

As such, we can find counterexamples to each of Quigley's Lemmas 5.11, 5.17, and 5.18, including counterexamples of arbitrary lengths to Quigley's Lemma 5.11. Thus, since Quigley's proof of the correctness of his algorithm relies on the implication discussed in these lemmas always being possible without processing clauses of length $k$ or more, then this proof is flawed. This means Quigley has not demonstrated that his algorithm correctly classifies 3CNF formulas. Thus, he has not demonstrated that $\p = \np$.

\subsection{A Counterexample to Quigley's Algorithm}\label{counterexample-base-case}
In the previous section, we showed that Quigley's argument is flawed. Now, we give a counterexample on which his algorithm fails.

Let $\phi$ be an unsatisfiable boolean formula in 4CNF form with $n$ clauses such that every clause contains exactly 4 terms and no clause in $\phi$ contains the same terminal more than once. Let $\Sigma_1$ denote the set of terminals that appear in $\phi$, and let $\Sigma_2$ denote the set of terms that appear in $\phi$. Now, we will construct a new boolean formula $\phi'$ in 3CNF form as follows:
\begin{enumerate}
    \item For each clause $c$ in $\phi$:
    \begin{enumerate}
        \item Let $a_1,a_2,a_3,a_4$ denote the first, second, third, and fourth terms in $c$, respectively.
        \item Construct two new clauses $(a_1\lor a_2\lor x_i)$ and $(a_3\lor a_4 \lor \overline{x_i})$, where $x_i$ is a new terminal that does not appear in $\phi$ and $i$ is the index of the clause $c$ in $\phi$.
    \end{enumerate}
    \item Take the conjunction of the clauses constructed in step 1.b. Let $\phi'$ be the 3CNF formula created this way.
\end{enumerate}

Note that, by this construction, $\phi'$ is of the form
\begin{equation*}
    (a_{1,1}\lor a_{1,2}\lor x_1) \land(a_{1,3}\lor a_{1,4}\lor \overline{x_1})\land\dots\land(a_{n,1}\lor a_{n,2}\lor x_n)\land(a_{n,3}\lor a_{n,4}\lor\overline{x_n})
\end{equation*}
where for all $1\leq i\leq n$, the terms $a_{i,1}$, $a_{i,2}$, $a_{i,3}$, and $a_{i,4}$ are all distinct terminals. Notice that the new third terminal $x_i$ is unique to the two clauses constructed from each original clause in $\phi$, so there are only two clauses in $\phi'$ that contain the terminal $x_i$, one in which $x_i$ is positive and another in which it is negated. This means there is exactly one clause in $\phi'$ containing the term $x_i$ and exactly one clause in $\phi'$ containing the term $\overline{x_i}$. Further, these two clauses that share the terminal $x_i$ cannot share any other terminals because the original clause they were constructed from contains no duplicate terminals by assumption. We denote the new third term $x_i$ or $\overline{x_i}$ of any clause $c \in \phi'$ as $x_c$. Let $X_1$ be the set of all terminals that appear in the new third term of some clause in $\phi'$, and let $X_2$ be the set of new third terms.

Note that any pair of adjacent clauses that share the same new third terminal implies a clause of length 4. This is because any two clauses $c_1$ and $c_2$ in $\phi'$ that share some terminal $x \in X_1$ cannot share any other terminals since $c_1$ and $c_2$ must have been constructed from some clause in the original formula $\phi$ and this original clause cannot have contained any duplicate terminals by assumption. Thus, $c_1$ is of the form $(a_{i,1}\lor a_{i,2}\lor x_i)$ and $c_2$ is of the form $(a_{i,3}\lor a_{i,4}\lor\overline{x_i})$ for some $1\leq i\leq n$, where the terms $a_{i,1}, a_{i,2}, a_{i,3}, a_{i,4}$ contain no duplicate terminals by assumption, meaning that $c_1$ and $c_2$ imply the clause $(a_{i,1}\lor a_{i,2}\lor a_{i,3}\lor a_{i,4})$, which is of length 4.

Now, we will show that $\phi'$ is unsatisfiable.

\begin{lemma}\label{phi'-unsat}
    The 3CNF formula $\phi'$ constructed as described previously is unsatisfiable.
\end{lemma}
\begin{proof}
    Recall that $\phi$ is unsatisfiable, so for each possible complete assignment $A$ to the variables in $\phi$, there exists some clause $c$ in $\phi$ that evaluates to false under that assignment. Since $c$ evaluates to false and $c$ is a disjunction of 4 terms, each of those terms must evaluate to false. Now, let $c_1$ and $c_2$ be the two clauses in $\phi'$ that are constructed from $c$ during step 1.b of the construction of $\phi'$. By definition, $c_1$ and $c_2$ each contains half the terms of $c$, all of which evaluate to false under the partial assignment $A$, along with a new terminal $x_i \in X_1$, which appears positive in one of the clauses $c_1$ and $c_2$ and negated in the other. Without loss of generality, suppose $x_i$ is positive in $c_1$ and negated in $c_2$. Then, under the partial assignment $A$, $c_1$ evaluates to $(F\lor\ldots\lor F\lor x_i)$, and $c_2$ evaluates to $(F\lor\ldots\lor F \lor\overline{x_i})$. Now, under any complete assignment to the variables in $\phi'$, $x_i$ must evaluate to either true or false. If $x_i$ is true, then $\overline{x_i}$ is false, so $c_2$ evaluates to $(F\lor\ldots\lor F)$, which is simply false; and if $x_i$ is false, then $c_1$ evaluates to $(F\lor\ldots\lor F)$, which is false. As such, any assignment to the variable $x_i$, and thus any complete assignment to the variables of $\phi'$, results in at least one clause in $\phi'$ evaluating to false, so $\phi'$ is unsatisfiable.
\end{proof}

However, we will now show that Quigley's algorithm classifies $\phi'$ as satisfiable.

\begin{theorem}\label{base-case-theorem}
    The 3CNF formula $\phi'$ is classified as satisfiable by Quigley's algorithm, as described in Section 6 of his paper \cite{qui:poly-sat}. Thus, Quigley's algorithm fails on $\phi'$.
\end{theorem}
\begin{proof}
    Recall that, by construction, $\phi'$ is of the form
    \begin{equation*}
        (a_{1,1}\lor a_{1,2}\lor x_1) \land(a_{1,3}\lor a_{1,4}\lor\overline{x_1})\land\dots\land(a_{n,1}\lor a_{n,2}\lor x_n)\land(a_{n,3}\lor a_{n,4}\lor\overline{x_n}).
    \end{equation*}
    Let $A$ denote the set of clauses in $\phi'$. During the first step of Quigley's algorithm, we iterate through all pairs of clauses $C,D\in A$ and check for any new clauses they imply according to Quigley's Lemma $5.9$ \cite{qui:poly-sat}. Recall that Quigley's algorithm ignores any clauses with length greater than 3 \cite{qui:poly-sat}, so we assume that implied clauses of length 4 or greater are not added to the instance. (Note that this counterexample will also hold without this assumption because even if clauses of length 4 or greater are added to the instance, they are never used to imply any further clauses. Thus, the clauses implied during the execution of the algorithm are the same with or without this assumption.)
    
    During this first iteration of the first step of the algorithm, no pair of clauses in $A$ can imply a clause of length either zero or one, since by Quigley's Lemma 5.10 \cite{qui:poly-sat}, the smallest clause that can be implied by two clauses of length 3 is a clause of length 2. Additionally, no pair of clauses can imply a clause of length exactly two in the first iteration of Quigley's algorithm \cite{qui:poly-sat}, as this would require both clauses to share all three of their terminals with each other. If two clauses do not share the same three terminals, they cannot imply a new clause of length 2 since there must be at least 4 unique terminals among them, only one of which will be removed from the new clause during implication, resulting in at least 3 distinct terminals in the final clause and thus a clause length of at least 3. However, for any $i$ with $1 \leq i \leq n$, the only two clauses in $A$ that share their third terminals $x_i\in X_1$ are the two clauses $c_i$ and $c_i'$ containing $x_i$ and $\overline{x_i}$, respectively, that are constructed from some 4-clause in $\phi$; since this 4-clause cannot contain any duplicate terms by assumption, it follows that $c_i$ and $c_i'$ also cannot share any terminals besides $x_i$, meaning they can only imply a clause of length 4, as described earlier. Further, any two clauses in $A$ that share first or second terminals cannot share their third terminals by construction since these clauses would have to have been constructed from different 4-clauses in the original formula $\phi$, and third terminals added during construction are unique to each original 4-clause. Therefore, no clauses of length 2 will be implied during the first iteration of the algorithm.

    The only case in which two clauses $C$ and $D$ imply a clause of length exactly 3 during the first iteration of the algorithm is when $C$ and $D$ share the same first two terminals $a$ and $b$ and have different third terminals. This is because if $C$ and $D$ share their third terminal, then they are either identical (if the third terminal has the same sign in both) and can imply no further clauses, or they can only imply a clause of length 4 (if the third terminal has opposite signs in each), as shown previously. Then, in order for these two clauses to resolve to a new clause, one of $a$ and $b$ must have opposite signs in $C$ and $D$, and in order for the resulting clause to have length 3, the other must have the same sign in both. Without loss of generality, suppose that $a$ has the same sign in both $C$ and $D$ and that $b$ appears positive in $C$ and negated in $D$. Then, $C$ has the form $(a \lor b \lor x_C)$ for some term $x_C \in X_2$, and $D$ has the form $(a \lor \overline{b} \lor x_D)$ for some term $x_D \in X_2$. As explained previously, $x_C$ and $x_D$ must be distinct terminals by construction. Therefore, in this case, $C$ and $D$ imply the clause $(a\lor x_C\lor x_D)$, and this new clause is added to $\phi'$ at the end of step 1.a of Quigley's algorithm \cite{qui:poly-sat}. By our assumption stated in Section \ref{sec-analysis-of-alg}, this new clause is not iterated over until the second iteration of step 1 of Quigley's algorithm.

    Therefore, after the first iteration of the first step of the algorithm, $\phi'$ is the conjunction of all the clauses originally in $\phi'$ before starting the algorithm and all clauses of length 3 implied during the first iteration. Note that since there are no clauses of length less than 3, then no clauses are expanded as per step 1.b of Quigley's algorithm. Let $B$ denote the set of new clauses added to $\phi'$ during the first iteration of step 1 of the algorithm. Then, after the first iteration of the first step of the algorithm, $\phi'$ is of the form
    \begin{gather*}
        (a_{1,1}\lor a_{1,2}\lor x_1) \land(a_{1,3}\lor a_{1,4}\lor\overline{x_1})\land\dots\land(a_{n,1}\lor a_{n,2}\lor x_n)\land(a_{n,3}\lor a_{n,4}\lor\overline{x_n}) \land \bigwedge_{c \in B} c
    \end{gather*}
    Note that no two clauses in $B$ can share the same three terminals. To see this, consider any clause $b_1\in B$, which has the form $(a_i\lor x_j\lor x_k)$ for some $a_i\in \Sigma_2$ and $x_j,x_k\in X_2$. Note that the terms $x_j$ and $x_k$ do not contain the same terminal, as explained earlier. Since $b_1$ is a clause in $B$, there must be two clauses $c_1,c_2\in A$ that imply it.
    More specifically, $c_1$ is of the form $(a_i\lor a_m\lor x_j)$ and $c_2$ is of the form $(a_i \lor\overline{a_m}\lor x_k)$ for some $a_m\in \Sigma_2$. This is because no clause in $A$ contains more than one term in $X_2$, so $x_j$ and $x_k$ must be in separate clauses. Then, in order for $c_1$ and $c_2$ to resolve with one another, both must contain some other terminal $a_m$ that is positive in one of $c_1$ and $c_2$ and negated in the other.
    Further, in order for $a_i$ to be in $b_1$, at least one of $c_1$ and $c_2$ must contain $a_i$. However, without loss of generality, if $c_1$ were to contain $a_i$ but $c_2$ did not, then since $c_2$ must still have length 3, $c_2$ would have to contain some other distinct term $a_j$. In this case, $a_j$ would also be in $b_1$, which would be a contradiction. Thus, $c_1$ and $c_2$ have the forms $(a_i\lor a_m\lor x_j)$ and $(a_i \lor\overline{a_m}\lor x_k)$, respectively.
    Note that there exist two other clauses $c_3,c_4\in A$ that contain the terms $\overline{x_j}$ and $\overline{x_k}$, respectively. However, $a_i$ and $a_m$ cannot appear in either of them since, by definition, no two clauses in $A$ that share some terminal $x\in X_1$ can share any other terminal $a\in\Sigma_1$. Therefore, in order to have length 3, $c_3$ and $c_4$ must each contain at least one other terminal $a_p \in \Sigma_1$ and $a_q \in \Sigma_1$, respectively, which is not found in $c_1$ or $c_2$. Thus any clause implied by $c_3$ or $c_4$ will contain either $a_p$ or $a_q$, respectively, and will not share all three terminals with $b_1$. Additionally, any clause implied by some pair of clauses including at least one clause other than $c_1$, $c_2$, $c_3$, or $c_4$ will not contain both the terminals $x_j$ and $x_k$ and will thus not share all three terminals with $b_1$. This means that each clause in $B$ contains a unique set of 3 terminals, meaning no two clauses in $B$ share all three of their terminals.
    
    After completing its first iteration of the first step, the algorithm then moves to step two. Since no clauses of length 1 have been implied so far, the algorithm moves to step 3. Multiple clauses have been added to the instance $\phi'$ (i.e., all the clauses in $B$), so the algorithm returns to step 1 for its second iteration.

    Now, we will show that no further clauses of length 3 or less are implied during the second iteration of the first step of Quigley's algorithm. Recall that $A$ consists of all clauses originally in $\phi'$ at the start of the algorithm, meaning any clause of length 3 or less implied by a pair of clauses in $A$ has already been added to the instance at the end of the first iteration of the algorithm. In other words, any clause implied by two clauses in $A$ is already in $B$, so a pair of clauses in $A$ cannot imply a new clause during the second iteration of the first step of the algorithm. Therefore, we only need to consider whether a new clause can be implied from either a pair containing a clause in $A$ and a clause in $B$ or from a pair of clauses that are both in $B$.

    First, consider some two clauses $a\in A$ and $b\in B$. By construction, $a$ contains two terminals in $\Sigma_1$ and one terminal in $X_1$, whereas $b$ contains one terminal in $\Sigma_1$ and two terminals in $X_1$. Note that $a$ and $b$ cannot imply a clause of length 0 or 1 since both clauses are of length 3, and they cannot imply a clause of length 2 because they do not share the same 3 terminals (since no terminal is both in $\Sigma_1$ and $X_1$). In order for $a$ and $b$ to imply a clause of length 3, there must be exactly 2 terminals that appear in both $a$ and $b$. This is because as stated earlier, $a$ and $b$ cannot share all three of their terminals, and if there are less than 2 terminals shared between $a$ and $b$, then the clauses contain at least 5 distinct terminals in total (i.e., 2 unique terminals each and at most 1 shared terminal), only one of which is removed through resolution, resulting in an implied clause of length at least 4. Thus, $a$ and $b$ each have exactly one unique terminal, and they have two shared terminals.
    Since $a$ contains two terminals in $\Sigma_1$, one of these shared terminals must be in $\Sigma_1$. Similarly, since $b$ contains two terminals in $X_1$, one of these shared terminals must be in $X_1$. Therefore, the pair of terminals that appear in both $a$ and $b$ consists of one terminal in $\Sigma_1$ and one terminal in $X_1$. In other words, there must exist some terminals $a_b \in \Sigma_1$ and $x_b \in X_1$ that are found in both $a$ and $b$. Then, since the remaining terminal in $a$ is some $a_i \in \Sigma_1$ and the remaining terminal in $b$ is some $x_j \in X_1$, these remaining terminals cannot be shared between $a$ and $b$.
    Now, in order for $a$ and $b$ to resolve with one another, one of their shared terminals must have opposite signs in each of $a$ and $b$. Then, in order for the clause resulting from the resolution to have length 3, the other shared terminal must have the same sign in $a$ and $b$. This is because if this other shared terminal were to also have opposite signs in $a$ and $b$, then there would be 4 distinct terms in the resulting implied clause: the term with $a_i$, the term with $x_j$, the term with the shared terminal in positive form, and the term with the shared terminal in negated form.
    Thus, exactly one of $a_b$ and $x_b$ must appear with opposite signs in $a$ and $b$, and the other must appear with the same sign.

    Now, notice that there are only two clauses that contain the terminal $x_b$ in $A$, and only one of these can contain $a_b$ (since clauses in $A$ that share a terminal in $X_1$ can share no other terminals). Thus, the only clause in $A$ that contains both $a_b$ and $x_b$ is $a$. By the construction of clauses in $B$ described earlier, since $b$ contains the terminals $a_b$, $x_b$, and $x_j$, then $b$ must have been implied during the first iteration of the first step of Quigley's algorithm by some clause containing $a_b$ and $x_b$ and some other clause containing $a_b$ and $x_j$ such that $a_b$ has the same sign in both. Since $a$ is the only clause containing both $a_b$ and $x_b$, it follows that $b$ must have been implied by $a$ and some other clause $d \in A$ that contains $a_b$ and $x_j$. It follows from the rules of resolution that since $a_b$ is found in both $a$ and $b$ and since $b$ is implied by $a$ and some other clause, then $a_b$ must appear with the same sign in both $a$ and $b$. This is because the term containing $a_b$ in $a$ must be the same term containing $a_b$ that is present in $b$. Likewise, since $x_b$ is found in both $a$ and $b$ and since $b$ is implied by $a$ and some other clause, then $x_b$ must appear with the same sign in both $a$ and $b$. However, this is a contradiction because as shown earlier, in order for $a$ and $b$ to resolve to a new clause of length 3, exactly one of $a_b$ and $x_b$ must have opposite signs in $a$ and $b$, and the other must have the same sign. Thus, $a$ and $b$ cannot resolve to a new clause of length 3 or less.

    Next, consider some two distinct clauses $b_1,b_2\in B$. By construction, $b_1$ and $b_2$ must both contain exactly one terminal in $\Sigma_1$ and two terminals in $X_1$. By the same reasoning used for $a$ and $b$ in the previous two paragraphs, $b_1$ and $b_2$ cannot imply a clause of length 0 or 1, and since no two distinct clauses in $B$ share all 3 terminals with each other, $b_1$ and $b_2$ cannot imply a clause of length 2. Then, in order for $b_1$ and $b_2$ to imply another clause of length 3, they must share exactly 2 terminals, and exactly one of these terminals must appear with the same sign in both clauses. Note that since $b_1$ and $b_2$ both contain one terminal in $\Sigma_1$ and two terminals in $X_1$, at least one of the two shared terminals must be a terminal in $X_1$. The other terminal can be in either $\Sigma_1$ or $X_1$, so we must consider both cases.
    
    \begin{description}
        \item[Case 1 (only one shared terminal is in $X_1$):] Suppose $b_1$ and $b_2$ share the terminals $a_b\in\Sigma_1$ and $x_b\in X_1$. Since only one clause $c\in A$ can contain both terminals $a_b$ and $x_b$ by definition, then $b_1$ and $b_2$ must each be implied by $c$. Also, since $c$ contains both $a_b$ and $x_b$, the terminals $a_b$ and $x_b$ must appear in $b_1$ and $b_2$ with the same sign as in $c$. Therefore, $a_b$ and $x_b$ appear with the same sign in both $b_1$ and $b_2$. This is a contradiction because in order for $b_1$ and $b_2$ to imply a new clause of length 3, exactly one of the two shared terminals $a_b$ and $x_b$ must appear with opposite signs in $b_1$ and $b_2$, and the other must appear with the same sign. Thus, $b_1$ and $b_2$ cannot imply any new clauses of length 3.

        \item[Case 2 (both shared terminals are in $X_1$):] Now, suppose instead that $b_1$ and $b_2$ share two terminals $x_{b_1},x_{b_2}\in X_1$. There are two clauses $c_1,c_2\in A$ that contain the terminal $x_{b_1}$ and two other clauses $c_3,c_4\in A$ that contain the terminal $x_{b_2}$. Therefore, since $b_1$ and $b_2$ each contain both $x_{b_1}$ and $x_{b_2}$, then they must each be implied by one of either $c_1$ or $c_2$ and one of either $c_3$ or $c_4$. Without loss of generality, suppose $c_1$ and $c_3$ imply $b_1$. Note that since $b_1$ and $b_2$ are distinct by assumption, then $c_1$ and $c_3$ cannot also imply $b_2$ because there is only one terminal that appears in both $c_1$ and $c_3$ but only appears negated in one of these clauses, meaning there is only one clause that $c_1$ and $c_3$ can imply.
        
        Additionally, consider $c_1$ and $c_4$. By construction, $c_1$ and $c_4$ do not share their third terminal since $c_1$ contains $x_{b_1}$ but not $x_{b_2}$ and $c_4$ contains $x_{b_2}$ but not $x_{b_1}$. Further, since $b_1 \in B$, then by construction, in order for $c_1$ and $c_3$ to imply $b_1$, $c_1$ and $c_3$ must share their first two terminals. Since $c_3$ and $c_4$ share their third terminal $x_{b_2}$, then they cannot share any other terminals, meaning they cannot share their first two terminals. It follows that $c_1$ and $c_4$ cannot share their first two terminals. Therefore, $c_1$ and $c_4$ do not share any terminals, so they cannot imply any clause, meaning they cannot imply $b_2$. For a similar reason, $c_2$ and $c_3$ cannot share any of their terminals, so they also cannot imply $b_2$.
        
        Therefore, $b_2$ must be implied by $c_2$ and $c_4$. However, recall that $x_{b_1}$ appears with opposite signs in $c_1$ and $c_2$, so $x_{b_1}$ must appear with opposite signs in $b_1$ and $b_2$. Similarly, $x_{b_2}$ appears with opposite signs in $c_3$ and $c_4$, so $x_{b_2}$ must appear with opposite signs in $b_1$ and $b_2$. Therefore, $b_1$ and $b_2$ cannot imply any new clauses of length 3, since neither of their two shared terminals appear with the same sign in both clauses.
    \end{description}

    As such, during the second iteration of step $1$ of Quigley's algorithm, no pair of clauses implies a new clause of length $3$ or less.
    
    After this second iteration of step $1$, Quigley's algorithm returns to step $2$. Since no clauses of length $1$ were added in step $1$, there are no contradicting clauses of length 1. Thus, the algorithm moves to step 3. No clauses of length 3 or less were added during the second iteration of step 1, so the algorithm moves to step $4$, which terminates the algorithm and classifies $\phi'$ as satisfiable. However, recall that $\phi'$ is unsatisfiable by Lemma \ref{phi'-unsat}. Therefore, Quigley's algorithm fails on $\phi'$.
\end{proof}

Thus, Quigley's algorithm does not return the correct result on the boolean formula $\phi'$. Since $\phi'$ is constructed from almost any of the infinite number of unsatisfiable boolean formulas in 4CNF form with only a few restrictions, it is clear that there are an infinite number of 3CNF formulas that can be generated in such a way and therefore an infinite number of counterexamples to Quigley's algorithm.

\subsection{Extension of Counterexample}\label{extension-of-counterexample}

In Section~\ref{counterexample-base-case}, we provide an unsatisfiable 3CNF formula that is incorrectly classified as satisfiable by Quigley's algorithm. In this section, we show that we can extend the method of construction of this counterexample to arbitrary lengths to generate counterexamples with different structures.

Let $\{ b_n \}$ be a sequence described by $b_0 = 3$ and $b_n = 2 (b_{n-1} - 1)$. Thus, for $n \in \naturalnumberzero$, the closed form expression is $b_n = 2^n + 2$. We can see this holds since if $n = 0$, then
\[ b_0 = 3 = 1 + 2 = 2^0 + 2 = 2^n + 2, \]
and if $b_n = 2^n + 2$ for some $n \geq 0$, then \[ b_{n+1} = 2(b_n - 1) = 2(2^n + 2 - 1) = 2(2^n + 1) = 2^{n+1} + 2. \]
The first several terms of this sequence are $3, 4, 6, 10, 18, ...$.

Consider some element $b_k$ of this sequence with $k \in \naturalnumberpositive$. Let $\phi_k$ be any unsatisfiable boolean formula in $b_k$CNF form with $n$ clauses such that each clause in $\phi_k$ contains exactly $b_k$ terms and no clause in $\phi_k$ contains the same terminal more than once. We construct a new boolean formula $\phi_{k-1}$ in $b_{k-1}$CNF form (that is, $\left(\frac{b_k}{2}+1\right)$CNF form) as follows.
\begin{enumerate}
    \item For each clause $c$ in $\phi_k$:
    \begin{enumerate}
        \item Let $a_1,...,a_{b_k}$ denote the terms in $c$.
        \item Construct two new clauses $(a_1 \lor \cdots \lor a_{b_k/2} \lor x_i)$ and $(a_{b_k/2 + 1} \lor \cdots \lor a_{b_k} \lor \overline{x_i})$, where $x_i$ is a new terminal that does not appear in $\phi_k$ and $i$ is the index of the clause $c$ in $\phi_k$.
    \end{enumerate}
    \item Let $\phi_{k-1}$ be the formula created by taking the conjunction of the clauses generated in step 1.b.
\end{enumerate}

By construction, this algorithm converts the original formula in $b_k$CNF form to a new formula in $b_{k-1}$CNF form. This is because for $k \in \naturalnumberpositive$, $b_k = 2(b_{k-1} - 1)$. Thus, when each clause is split in half in step 1.b, each of the two new resulting clauses has $\frac{b_k}{2} = \frac{2(b_{k-1} - 1)}{2} = b_{k-1} - 1$ of the original clause's terms, and when a new terminal $x_i$ or $\overline{x_i}$ is added to each of the two new clauses, each of the two resulting clauses has length $(b_{k-1} - 1) + 1 = b_{k-1}$.

By repeatedly applying this algorithm to the original $b_k$CNF formula $\phi_k$, we get a new $b_0$CNF formula $\phi_0$. Since $b_0=3$, $\phi_0$ is a boolean formula in $3$CNF form.
We will show $\phi_0$ is unsatisfiable but is classified as satisfiable by Quigley's algorithm.

First, we will show $\phi_0$ is unsatisfiable.

\begin{lemma}\label{resulting-formulas-unsat}
    Given a boolean formula $\phi_k$ in $b_k$CNF form for some $k\in\naturalnumberpositive$ meeting the constraints described earlier, all boolean formulas in $b_j$CNF form for $0\leq j < k$ created by repeatedly applying the procedure described are unsatisfiable.
\end{lemma}
\begin{proof}
    We proceed by induction over $k$.

    First, consider the boolean formula $\phi_k$ in $b_k$CNF form. By construction, this formula is unsatisfiable.

    Next, suppose the boolean formula $\phi_j$ in $b_j$CNF form is unsatisfiable for some $j\in\naturalnumberpositive$ such that $1\leq j \leq k$. We will show the boolean formula $\phi_{j-1}$ in $b_{j-1}$CNF form constructed as described previously is also unsatisfiable. Since $\phi_j$ is unsatisfiable, then for every complete assignment $A$ to the terminals in $\phi_j$, there must be at least one clause $c$ in $\phi_j$ that evaluates to false. Since each clause in $\phi_j$ is a disjunction of $b_j$ terms, each term in $c$ must evaluate to false under $A$. Thus, $c$ evaluates to the clause $(F\lor \ldots \lor F)$ under the complete assignment $A$.
    Now, let $c_1$ and $c_2$ be the two clauses in $\phi_{j-1}$ that are constructed from $c$ during step 1.b of the procedure described previously. Then, by construction, $c_1$ and $c_2$ each contains half the terms of $c$, all of which evaluate to false under the partial assignment $A$, along with a new terminal $x_i$, which is positive in one of the clauses $c_1$ and $c_2$ and negated in the other. Without loss of generality, suppose $x_i$ is positive in $c_1$ and negated in $c_2$. Then, under the partial assignment $A$, $c_1$ evaluates to $(F\lor\ldots\lor F\lor x_i)$, and $c_2$ evaluates to $(F\lor\ldots\lor F \lor\overline{x_i})$. Now, under any complete assignment to the variables in $\phi_{j-1}$, $x_i$ must evaluate to either true or false. If $x_i$ is true, then $\overline{x_i}$ is false, so $c_2$ evaluates to $(F\lor\ldots\lor F)$, which is simply false; and if $x_i$ is false, then $c_1$ evaluates to $(F\lor\ldots\lor F)$, which is false. As such, any assignment to the variable $x_i$, and thus any complete assignment to the variables of $\phi_{j-1}$, results in at least one clause in $\phi_{j-1}$ evaluating to false, so $\phi_{j-1}$ is unsatisfiable.
\end{proof}

Since $\phi_0$ is constructed by repeatedly applying the algorithm described previously to a boolean formula in $b_k$CNF form for some $k\in\naturalnumberpositive$, then by Lemma \ref{resulting-formulas-unsat}, $\phi_0$ must be unsatisfiable. However, we will now show that Quigley's algorithm classifies $\phi_0$ as satisfiable.

\begin{theorem}\label{inductive-case-theorem}
    For any $k \in \naturalnumberpositive$, a 3CNF formula $\phi_0$ constructed from an unsatisfiable $b_k$CNF formula $\phi_k$ as described previously is classified as satisfiable by Quigley's algorithm, as described in Section 6 of his paper \cite{qui:poly-sat}. Thus, Quigley's algorithm fails on $\phi_0$.
\end{theorem}
\begin{proof}
    We proceed by induction over $k$.

    Consider the base case when $k = 1$. Then, $b_k = b_1 = 2^1 + 2 = 4$. By Theorem~\ref{base-case-theorem}, the formula is unsatisfiable but classified as satisfiable by Quigley's algorithm.

    Now, consider the inductive case. In particular, assume that for some $k \in \naturalnumberpositive$, a 3CNF formula constructed from any unsatisfiable $b_k$CNF formula by the procedure given earlier is classified as satisfiable by Quigley's algorithm. We will show the same holds for a 3CNF formula constructed from an unsatisfiable $b_{k+1}$CNF formula. Consider some unsatisfiable $b_{k+1}$CNF formula $\phi_{k+1}$ with $n$ clauses such that each clause in $\phi_{k+1}$ contains exactly $b_{k+1}$ terms and no clause in $\phi_{k+1}$ contains the same terminal more than once. In the first iteration of the procedure for constructing $\phi_0$, we construct a $b_k$CNF formula $\phi_k$ from $\phi_{k+1}$. By Lemma~\ref{resulting-formulas-unsat}, since $\phi_{k+1}$ is unsatisfiable, then $\phi_k$ is unsatisfiable. Additionally, by construction, $\phi_k$ is a $b_k$CNF formula in which no clause contains the same terminal more than once. This is because in each clause of $\phi_k$, none of the first $b_k-1$ terminals in the clause can be repeated since they come directly from a $b_{k+1}$-clause in which no terminals are repeated, and the new last terminal $x_i$ is created such that it is a terminal not previously found in the formula, so it cannot be any of the first $b_k-1$ terminals. Then, since constructing $\phi_k$ is part of the procedure for constructing the 3CNF formula $\phi_0$ from $\phi_{k+1}$, then if we were to start the procedure with $\phi_k$, the 3CNF formula constructed from $\phi_k$ by the procedure will also be precisely $\phi_0$. Thus, by the inductive hypothesis, Quigley's algorithm classifies $\phi_0$ as satisfiable.

    Therefore, although $\phi_0$ is unsatisfiable, Quigley's algorithm classifies it as satisfiable, so Quigley's algorithm fails on $\phi_0$.
    
    The result follows by induction.
\end{proof}

\subsection{Analysis of Quigley's Algorithm With No Bounds}

Notice that a key reason why Quigley's algorithm fails on the given example instances is that in each iteration of the algorithm, it only attempts to resolve or expand clauses of length at most 3. This gives the algorithm a polynomial runtime (given the assumptions stated in Section \ref{section:runtime}) but results in the algorithm producing incorrect results on some inputs. Thus, a natural question to consider is whether one can increase or remove the bound on the length of clauses considered by step $1$ of Quigley's algorithm. If one were to increase the bound, then Quigley's justification of the correctness of his algorithm still fails since a counterexample to his Lemma 5.11 exists for clauses of any length.

However, if one were to remove the bound altogether, then the algorithm would no longer be guaranteed to run in polynomial time. In fact, this change would cause Quigley's algorithm to require exponential time and space in some cases. To see why, consider the boolean formula $x_1\land x_2\land \ldots\land x_n$ with $n$ clauses for any $n \in \naturalnumberpositive$ such that $x_1, ..., x_n$ are distinct variables. Clearly, this formula is in 3CNF form as each clause has length 1, and it is satisfiable by assigning each terminal the value true. During the first iteration of step $1$ of Quigley's algorithm, the algorithm will check whether the length-1 clause $x_1$ alongside any other clause can imply any new clauses~\cite{qui:poly-sat}. By construction, no new clauses will be implied since none of the clauses can resolve with one another. In step 1.b of the algorithm, $x_1$ will be expanded to create all clauses it implies according to Quigley's Lemma $5.8$~\cite{qui:poly-sat}. Since there is no bound on clause length, the maximum length of a clause implied this way is $n$, in which case the implied clause would contain every terminal that appears in the formula exactly once. Since each terminal other than $x_1$ can appear in either positive or negated form in these new clauses, the clause $x_1$ implies $2^{n-1}$ new clauses according to Quigley's Lemma $5.8$~\cite{qui:poly-sat}. Thus, $O(2^n)$ new clauses are added to the formula during the first iteration of step $1$ of the algorithm, meaning that it fails to run with polynomial space and therefore fails to run in polynomial time.

\section{Conclusion}

In this paper, we have shown that Quigley's algorithm is flawed by providing counterexamples to several lemmas with which Quigley claims to prove the correctness of his algorithm. We have also provided an infinite set of unsatisfiable boolean formulas that are incorrectly classified as satisfiable by Quigley's algorithm and have shown that removing the bound on clause length causes the algorithm to require exponential space and time in the worst case. As a result, Quigley fails to provide an algorithm that decides 3SAT in deterministic polynomial time, and he thus fails to prove that $\p=\np$.

\paragraph{Acknowledgments}
We would like to thank Lane~A.~Hemaspaandra and Michael~Reidy for their helpful comments on prior drafts. The authors are responsible for any remaining errors.

\bibliographystyle{alpha}
\bibliography{gry-reu,local_refs}

\end{document}